\documentclass[manuscript,screen,nonacm]{acmart}

\AtBeginDocument{%
  }

\setcopyright{acmlicensed}
\copyrightyear{2025}
\acmYear{2025}
\acmDOI{XXXXXXX.XXXXXXX}

\acmJournal{TAISAP}

\usepackage[ruled,linesnumbered]{algorithm2e}
\usepackage{graphicx}  
\usepackage{multirow}   
\begin{document}

\title{A Versatile Framework for Designing Group-Sparse Adversarial Attacks}

\author{Alireza Heshmati}
\email{alireza.heshmati@ee.sharif.edu}
\orcid{0000-0001-6023-7964}
\author{Saman Soleimani Roudi}
\email{saman.soleimani@ee.sharif.edu}
\orcid{0009-0002-9095-9313}
\affiliation{%
	\department{Department of Electrical Engineering}
	\institution{Sharif University of Technology}
	\city{Tehran}
	\country{Iran}
}

\author{Sajjad Amini}
\email{s\_amini@sharif.edu}
\orcid{0000-0002-0322-9324}
\author{Shahrokh Ghaemmaghami}
\email{ghaemmag@sharif.edu}
\orcid{0000-0002-9556-8090}
\affiliation{
	\department{Electronics Research Institute (ERI)}
	\institution{Sharif University of Technology}
	\city{Tehran}
	\country{Iran}
}

\author{Farokh Marvasti}
\email{marvasti@sharif.edu}
\orcid{0000-0002-4635-8986}
\affiliation{%
	\department{Department of Electrical Engineering}
	\institution{Sharif University of Technology}
	\city{Tehran}
	\country{Iran}
}

\renewcommand{\shortauthors}{Heshmati et al.}

\begin{abstract}
Existing adversarial attacks often overlook the sparsity of perturbations, limiting their effectiveness in modeling structural perturbations and explaining how deep neural networks (DNNs) process meaningful input patterns.
In this paper, we propose our framework ATOS (Attack Through Overlapping Sparsity), which leverages an efficient differentiable optimization framework to generate various forms of structured sparse adversarial perturbations, including element-wise, pixel-wise, and group-wise types.
For white-box attacks on image classifiers, we introduce the Overlapping Smoothed $\ell_0$ (OSL0) function, which ensures convergence to a stationary point while making the perturbations sparse and structured. By grouping channels and adjacent pixels, ATOS enhances the interpretability of adversarial examples and facilitates the identification of robust versus non-robust image features. Furthermore, we utilize the logarithm of the sum of exponential absolute values to approximate the $\ell_\infty$ gradient, enabling tight control over perturbation magnitude. ATOS achieves a $100\%$ attack success rate on benchmark datasets such as CIFAR-10 and ImageNet, while significantly improving the sparsity and structural coherence of perturbations compared to state-of-the-art techniques. Furthermore, the structured group-wise attack of ATOS highlights critical regions from the DNN’s perspective, offering improved explainability and serving as a counterfactual explanation by replacing key regions of the original class with robust features from the target class.
\end{abstract}






\maketitle
\section{Introduction}
Deep neural networks (DNNs) have shown remarkable performance in various machine learning domains, particularly computer vision \cite{dlfwc20, buric2025next, yang2025diffmic,he2025diffusion, awais2025foundation}. Despite their glamorous accuracy, DNNs are highly vulnerable to adversarial perturbations generated by adversarial attacks, which can mislead models and may even be imperceptible to humans \cite{szegedy2014intriguing, aeozl19,long2022survey, zhang2025adversarial}. 

This vulnerability stems from the fact that standard DNNs tend to rely on non-robust features, noise-like patterns that are sensitive to small changes, rather than robust features, which align more closely with natural data and resist perturbations \cite{ilyas2019adversarial, aaola20}. Although standard training uses non-robust features to improve the training accuracy of DNNs, these features increase susceptibility to adversarial attacks \cite{ilyas2019adversarial, li2024adversarial}. In contrast, robust DNNs, typically achieved through adversarial training \cite{madry2017towards}, emphasize robust features, causing most adversarial attacks to fail or generate semantically meaningful and perceptible perturbations \cite{jeanneret2023adversarial}. However, generating such perturbations is not the primary goal of adversarial attacks \cite{szegedy2014intriguing, aaola20}. Instead, these perturbations can serve as counterfactual explanations, where the closest meaningful image in terms of a distance metric, such as sparsity, alters the predicted class of the model \cite{goyal2019counterfactual, jeanneret2023adversarial,carrizosa2024mathematical}. Consequently, strong adversarial attacks, whose design is the main aim of this paper, are valuable for evaluating DNN robustness, improving standard DNNs through adversarial training, and generating counterfactual explanations for robust models.

Adversarial attacks aim to transform an input image into an adversarial example by adding perturbations, fooling trained DNNs while both share significant features and the main class remains perceptible to humans. Depending on the attacker's knowledge of the DNN, adversarial attacks are classified into two types: white-box, where the attacker knows the model's architecture and parameters, and black-box, where only the model outputs are accessible. While black-box attacks are useful for fooling inaccessible DNNs, white-box attacks are valuable for analyzing and robustly training specific DNNs. Therefore, this paper aims to design a white-box attack. These attacks can also be classified as targeted, aiming for a specific output, or untargeted, where the goal is simply to mislead the model \cite{akhtar2018threat,tang2024black, li2024adversarial}.

The general formulation of a white-box attack for image classification is:
\begin{equation}
	\min_{\boldsymbol{\Delta}} \; \mathcal{L}(\boldsymbol{y}_{\delta},\boldsymbol{y}_t)
	+ \mathcal{D}(\boldsymbol{X},\boldsymbol{X}+\boldsymbol{\Delta})
	\quad\text{s.t.}\;
	\boldsymbol{X}+\boldsymbol{\Delta} \in [0,1]_q^{\,c\times w \times h},
	\label{eq-atce}
\end{equation}
Here, $\boldsymbol{X}$ denotes the input image and $\boldsymbol{\Delta}$ denotes the adversarial perturbation. The function $\mathcal{D}(\cdot,\cdot)$ is a distance metric measuring perturbation size, and $\mathcal{L}(\cdot,\cdot)$ is the objective (loss) used to achieve the attack goal. The model output on the perturbed image $\boldsymbol{X}+\boldsymbol{\Delta}$ is written $\boldsymbol{y}_{\delta}$, while $\boldsymbol{y}_{t}$ denotes the adversary's target output. The tuple $(c,w,h)$ indicates the image shape (channels, width, height). The index $q$ denotes quantization: the adversarial pattern is quantized to the original pixel-value scale, typically the integer set $\{0,1,\ldots,255\}$, before being clamped to the valid input range.

The distance metric in formulation \eqref{eq-atce}  ensures that the adversarial image remains close to the input image by applying regularization, allowing the main class to remain perceptible to humans and preventing the attack from creating arbitrary perturbations. The most used distance metric constrains the value of each perturbation element, making it imperceptible to humans \cite{goodfellow2014explaining, croce2020reliable}. 
However, these attacks do not focus on perturbation sparsity and thus do not aim to minimize the number of perturbed items in the input image, while sparsity regularization is useful for identifying fragile items in the input image \cite{ilyas2019adversarial}.
In this context, pioneering work by \cite{opajd19} has used evolutionary algorithms to illustrate the vulnerability of deep neural networks to adversarial attacks by manipulating only a single pixel in the input image.
Since this regularization focuses solely on the sparsity of perturbations while neglecting their intensity, the resulting adversarial examples are often perceptible. Consequently, subsequent studies have relaxed the regularization on the number of perturbed items and aimed to develop algorithms that control both invisibility and sparsity simultaneously \cite{safmm19, saich19, amini2024fast}. 

In many attack designs, sparsity regularization treats each element of the pixel channels separately, applying element-wise sparsity regularization to limit the number of perturbed elements. However, since a real image is a two-dimensional array, the alteration of each pixel (corresponding to a change in at least one channel within that pixel) should be measured, rather than counting the number of altered channels across all pixels. This scenario can be effectively addressed through pixel-wise sparsity regularization, a group sparsity metric that differentiates between element-wise and pixel-wise sparsity by treating each pixel as a single unit, rather than focusing on individual color channels. This regularization is effective for maintaining invisibility due to the generally scattered perturbations and its focus on fewer perturbed pixels. However, it fails to form coherent regions that can highlight vulnerable areas from a DNN's perspective \cite{xu2018structured}. Consequently, incorporating group-wise perturbation in attack design becomes necessary, in addition to pixel-wise perturbation, which is either rare or weak in recent frameworks. Metrics that promote the perturbation of adjacent pixels together are crucial for targeting specific regions in the input image, enhancing attack explainability, and generating counterfactual examples \cite{jeanneret2023adversarial,ICLR2025_398b00a0}. Therefore, in our attack design, both grouping and sparsity play a significant role.

Group sparsity was first introduced for selecting grouped factors in regression \cite{yuan2006model}. Unlike traditional sparsity regularizers that focus on individual elements, group sparsity regularizers encourage sparse use of input element groups \cite{yuan2006model,hamidi2010fast}.
An extension of group sparsity regularization involves overlapping groups \cite{jenatton2011structured,jacob2009group}. While some modalities, such as image, have naturally group structures (e.g., each pixel comprises three color channels: Red, Green, and Blue), others may not have predefined groups. In structured adversarial perturbations, the attack must identify optimal groups in various positions and shapes, making overlapping groups advantageous for flexibility in non-trivial scenarios \cite{xu2018structured}.

In this paper, we propose a novel white-box adversarial attack on DNNs, called ATOS (Attack Through Overlapping Sparsity). We introduce a general sparsity function that supports all sparsity modes, including pixel-wise and group-wise, and use a practical function to approximate the gradient of the $\ell_\infty$ norm to regularize perturbation intensity. The proposed framework creates accurate adversarial examples and counterfactual explanations.

Our main contributions are:
\begin{itemize}
	\item \textbf{Unified sparse attack framework:} We propose ATOS, a versatile framework for generating element-, pixel-, or group-sparse adversarial perturbations via an overlapping sparsity regularizer, with provable convergence.
	
	\item \textbf{Adaptive intensity regularization:} Unlike threshold-based adversarial attack frameworks \cite{saich19,xu2018structured,zhu2021sparse,amini2024fast}, ATOS employs a gradient-based $\ell_\infty$ approximation that adapts perturbation strength per input, enabling imperceptible or controlled attacks without preset thresholds.

	\item \textbf{Explainable attack on DNNs:} ATOS offers superior sparsity control and reveals model vulnerabilities by targeting salient input patterns. Its group-wise attack with intensity regularization also provides counterfactual explanations for robust models—unavailable in prior frameworks.
	
\end{itemize}

The remainder of this paper is organized as follows. Section  \ref{nap} introduces the notations and preliminaries. Section \ref{pa} provides a review of previous frameworks. Section \ref{om} presents a detailed description of ATOS. Section \ref{sr} discusses the simulation results comparing ATOS with others. Finally, Section \ref{c} concludes the paper.


\section{Notations and Preliminaries} \label{nap}

In this paper, lowercase letters denote scalars and bold lowercase letters denote vectors. Matrices and tensors are denoted by bold uppercase letters. The scalar $v_i$ denotes the $i$-th element of the vector $\boldsymbol{v}$, while $x_{i,j}$ denotes the element in the $i$-th row and $j$-th column of the (2-D) tensor $\boldsymbol{X}$. The sign function is written $\operatorname{sgn}(\cdot)$. Vector norms — the $\ell_0$ norm (number of non-zero elements), the $\ell_1$ norm (sum of absolute values), the $\ell_2$ norm (Euclidean norm) and the $\ell_\infty$ norm (maximum absolute value) — are extended to matrices and tensors by applying the same norm to the vectorized form of them.

The categorizing operator is written
\[
\boldsymbol{x}=\mathcal{C}(\boldsymbol{X};\mathcal{R}),
\]
which converts tensor $\boldsymbol{X}$ into vector $\boldsymbol{x}$ according to a grouping (categorizing) rule $\mathcal{R}$. Simple vectorization is denoted
\[
\boldsymbol{x}=\mathcal{V}(\boldsymbol{X}).
\]
Other categorizing rules are introduced in Section~\ref{om}.

Below we recall two useful surrogates for the $\ell_0$ and $\ell_\infty$ norms, the definition of Lipschitz gradients, and a few elementary lemmas used later to guarantee convergence properties.

\begin{definition}[Smoothed $\ell_0$ norm \cite{mohimani2008fast}]
	\label{def_sl0}
	Let $\boldsymbol{x}\in\mathbb{R}^N$. The Smoothed $\ell_0$ (SL0) surrogate is defined as:
	\[
	\mathrm{S}_{\ell_0}(\boldsymbol{x}) = \sum_{i=0}^{N-1}\Big(1 - \exp\!\big(-x_i^2/(2\sigma^2)\big)\Big),
	\]
	where $\sigma>0$ is a relaxation parameter. For large $\sigma$ one has behavior similar to the squared $\ell_2$ norm (loosely, $\mathrm{S}_{\ell_0}(\boldsymbol{x}) \approx \|\boldsymbol{x}\|_2^2$ for $\sigma\gg 1$), while as $\sigma\to 0$ the surrogate approaches the true $\ell_0$ norm.
\end{definition}

\begin{definition}[LSEAp: approximation of the gradient of $\ell_\infty$ \cite{heshmati2022designing}]
	\label{def_LSEAp}
	Let $\boldsymbol{x}\in\mathbb{R}^N$, let $l_x=\|\boldsymbol{x}\|_\infty$, and let $p>0$ be a relaxation parameter. Define the LSEAp function (logarithm of the sum of exponential absolutes) by:
	\[
	\mathrm{LSEAp}(\boldsymbol{x};p) = \frac{1}{p}\,\ln\!\Big(\sum_{i=1}^N \exp\!\big(p(|x_i|-l_x)\big)\Big).
	\]
	An approximation to the gradient of the $\ell_\infty$ norm is then
	\[
	\frac{\partial}{\partial x_k}\mathrm{LSEAp}(\boldsymbol{x};p)
	=
	\operatorname{sgn}(x_k)\,\frac{\exp\!\big(p(|x_k|-l_x)\big)}{\sum_{i=1}^N \exp\!\big(p(|x_i|-l_x)\big)}.
	\]
	For $p\gg 1$ this formula concentrates on the maximal elements and approaches the gradient of $\|\boldsymbol{x}\|_\infty$, while numerically avoiding extremely large exponent values and improving stability and convergence speed in practice \cite{heshmati2022designing}.
\end{definition}

\begin{definition}[Lipschitz gradient {\cite{nesterov2018lectures}}]
	\label{def_l1lip}
	A function $f:\mathbb{R}^d\to\mathbb{R}$ is said to have an $L$-Lipschitz gradient if, for all $\boldsymbol{x},\boldsymbol{y}\in\mathbb{R}^d$,
	\[
	\|\nabla f(\boldsymbol{x}) - \nabla f(\boldsymbol{y})\|_2 \le L\,\|\boldsymbol{x}-\boldsymbol{y}\|_2.
	\]
	The smallest such constant $L$ is called the Lipschitz constant of the gradient. A first-order method with step size $\eta\le 1/L$ then satisfies the usual smooth-descent guarantees (and is a natural choice when proving convergence to stationary points for smooth objectives).
\end{definition}

\begin{lemma}
	\label{lemm_binomial}
	For any real numbers $a$ and $b$ it holds that
	\[
	a^2 + b^2\ge 2|a||b|,\qquad \forall\, a,b\in\mathbb{R},
	\]
	a direct consequence of $(|a|-|b|)^2\ge 0$.
\end{lemma}

\begin{lemma}[Bound on $a x^2 e^{-a x^2}$ for $a>0$]
	\label{lemm_bound_exp}
	Let $a>0$ and $g(x)=a x^2 e^{-a x^2}$. Then $g(x)$ is bounded on $\mathbb{R}$; its maximum is attained at $x=\pm 1/\sqrt{a}$ and one has
	\[
	0 \le g(x) \le e^{-1}.
	\]
\end{lemma}

\section{Prior Art} \label{pa}
One of the first adversarial attacks on DNNs was designed in 2014, minimizing changes to the input pattern using the $\ell_2$ norm criterion to increase the similarity between the adversarial and original images \cite{szegedy2014intriguing}. This framework uses the gradient through networks for its white-box attack. Subsequent adversarial attacks, such as the C\&W attack, introduced support for different norms, including $\ell_0$, $\ell_2$,  and $\ell_\infty$, to control the attack’s invisibility \cite{carlini2017towards}. The C\&W attack uses heuristic methods to approximate the $\ell_0$ and $\ell_\infty$ norms and also proposes alternative objective functions beyond Cross-Entropy for the first time.

Jacobian-based Saliency Map Attack (JSMA) creates a targeted sparse attack by evaluating the impact of each element on the DNN’s performance, based on a saliency map, which is used to manipulate key elements to achieve the target \cite{papernot2016limitations}.
SparseFool designs sparse adversarial attacks by employing a linear approximation of boundaries, promoting sparsity with the $\ell_1$ norm while controlling perturbation within image limits using the $\ell_2$ norm \cite{safmm19}. This method can remain imperceptible by defining a threshold for the amount of perturbation.
GreedyFool \cite{dong2020greedyfool} and CCVSAA (Color Channel Volterra expansion Sparse Adversarial Attack) \cite{cheng2024color} are additional examples of sparse adversarial attacks. These frameworks utilize gradient and distortion maps to identify effective elements while discarding less important ones through a heuristic search. CCVSAA further considers interactions among RGB channels via Volterra expansion \cite{flake1963volterra}. However, a key limitation of these frameworks is their exclusive focus on element-wise perturbations, without considering pixel-wise  perturbations.

Projected Gradient Descent (PGD) with $\ell_0$ regularization, referred to as PGD-$\ell_0$ \cite{saich19}, is a sparse adversarial attack based on the original PGD attack \cite{madry2017towards}. This framework can be made imperceptible (known as PGD-$\ell_0,\ell_\infty$) by setting a threshold for maximum absolute changes. It first finds an adversarial example with constraints on the squared $\ell_2$ norm and the threshold limit. Then, it measures the energy of each pixel perturbation and selects $k$ pixels with the maximum energy to fool DNNs. Unlike previous attacks, this is a pixel-wise attack that does not select individual channel elements separately. However, the number of perturbed pixels must be specified for this framework.

Gradual Sparse Attack (GSA) employs proximal-based optimization methods to design the attack \cite{amini2024fast}. GSA uses the $\ell_1$ norm, $\ell_0$ norm, and the Smoothly Clipped Absolute Deviation (SCAD) function \cite{vsvfl01} for sparsity regularization. It begins with a dense attack and gradually makes it sparse using the penalty method, which demonstrates superior performance in element-wise attacks compared to other related frameworks. This gradual approach inspires ATOS, which employs an alternative and more effective approximation of the $\ell_0$ norm. In addition, ATOS is designed to enable both pixel-wise and group-wise attacks, which cannot be achieved using the GSA framework.

Few studies have focused on group-wise attacks on DNNs. An early work introduced structured adversarial attacks (StrAttack) using the Alternating Direction Method of Multipliers (ADMM) to enforce group-wise sparsity through a dynamic sliding mask \cite{xu2018structured}. Subsequently, sparse adversarial attacks were developed using perturbation factorization, known as Homotopy, which was initially designed for sparse attacks and later extended to support group sparsity attacks \cite{zhu2021sparse}. Both of these attacks use a predefined threshold to limit the $\ell_\infty$ norm for invisibility. Later, FWnucl was proposed as a structured adversarial attack that leverages nuclear group norms to enhance attack's efficacy while maintaining invisibility \cite{kazemi2023minimally}. Most recently, Group-wise Sparse and Explainable adversarial attacks (GSE) were presented, utilizing an optimization framework that combines non-convex regularization with Nesterov’s accelerated gradient descent \cite{ICLR2025_398b00a0}. Based on our experience, which will be demonstrated in this paper, all these attacks either have a low Attack Success Rate (ASR) or cannot always generate structured perturbations.

In this paper, we propose a basic framework for designing pixel-wise and group-wise (structured) attacks, where the perturbation intensity can be controlled using an approximation of the $\ell_\infty$ norm. These attacks are largely automated and do not require manual definition of an invisibility threshold or specification of the number of sparse elements or pixels. This contrasts with previous adversarial attack frameworks which control the $\ell_\infty$ norm.
In addition, automated categorization is achieved by using overlapping windows.	

\section{Our Method} \label{om}
Based on \eqref{eq-atce}, the proposed framework for ATOS is formulated as:
\begin{align}
	\min_{\boldsymbol{\Delta}}\;  \mathrm{CE}(\boldsymbol{y}_{\delta},\boldsymbol{y}_t)
	+ \lambda_s\, r_s\big(\mathcal{C}(\boldsymbol{\Delta};\mathcal{R})\big)
	+ \lambda_{\infty}\, r_{\infty}\big(\mathcal{V}(\boldsymbol{\Delta})\big) & \quad
	\text{s.t.}\; \boldsymbol{X} + \boldsymbol{\Delta} \in [0,1]^{\,c\times w\times h}
	\label{eq-main-attack}
\end{align}

where $\lambda_s$ and $\lambda_\infty$ are hyperparameters that control the sparsity and invisibility of the perturbation, respectively. $ r_s(\cdot)$ is an Overlapping Smoothed $\ell_0$ (OSL0) regularization designed to cover all sparsity modes (element-wise, pixel-wise, and group-wise), with the categorizing function $\mathcal{C}(\cdot;\cdot)$. The rule $\mathcal{R}$ adapts the regularizer to different sparsity regularization modes. $\text{CE}(\cdot,\cdot)$ represents the cross-entropy loss. Thus, Problem \eqref{eq-main-attack} is a targeted attack requiring a fixed $\boldsymbol{y}_t$ in advance. By choosing $\boldsymbol{y}_t$ as the most probable label other than the main label in each optimization iteration, this problem can also be used for untargeted adversarial attacks.

Unlike other frameworks that use a threshold to constrain the $ \ell_\infty $ norm to reduce attack visibility, we propose the LSEAp function in Definition \ref{def_LSEAp}, denoted as $ r_\infty(\cdot) $, which automatically reduces the $ \ell_\infty $ norm without requiring a specific threshold for each image. It is a convex function with a relaxation parameter $p$, where increasing $p$ improves approximation but slows convergence. To balance this trade-off, $p$ can gradually increase during first-order optimization \cite{heshmati2022designing}. However, in ATOS, this issue is effectively addressed by using sufficient iterations. In addition, ATOS focuses on reducing the intensity of the attack for key points of images rather than directly limiting the $\ell_\infty$ norm. Moreover, since it is challenging to create an imperceptible adversarial example for robust networks, as shown in \cite{amini2024fast}, the LSEAp function helps generally control the amount of perturbation, similar to the $\ell_2$ norm, to create counterfactual explanations.

It is to be noted that in the constraint of \eqref{eq-atce}, there is quantization to map it into the original image scale ($[0,1,\ldots,255]$), but this is omitted in Problem \eqref{eq-main-attack} and  the quantization is done at the end of our optimization algorithm.

In the following, we first investigate OSL0 function which is the basis of our regularization term to cover sparsity and grouping of the perturbation. Next, we complete the algorithm to design element-wise, pixel-wise, and group-wise attacks with this regularization.

\subsection{Overlapping SL0}
Group sparsity with overlap is the most general concept considering correlations between groups, particularly those that are close to each other. In other words, overlapping group sparsity can be regarded as an adaptive group sparsity where the size and position of each group are variable, allowing it to find all possible sets of groupings. 
In this paper, we introduce the Overlap Smoothed $\ell_0$ (OSL0) function as a general approach to support all extensions of SL0, such as individual sparsity in Definition \ref{def_sl0}, group sparsity \cite{hamidi2010fast}, and overlapping group sparsity.

Assume the following grouping for the vectorized input $\boldsymbol{x}$:
\begin{equation}
	\label{ogrouping}
	\boldsymbol{x} =
	\big[\,\underbrace{x_0,\ldots,x_s,\ldots,x_{n_v-1}}_{\text{Group-1}}\;,\;
	\underbrace{x_s,\ldots,x_{s+n_v-1}}_{\text{Group-2}}\;,\;\ldots\,\big]^T,
\end{equation}
where $n_v$ is the number of elements in each group and $s$ is the overlap stride. The stride is used for grouping and $n_v - s$ represents the number of common elements for two sequential groups. 

Hence, the OSL0 function is defined as:
\begin{align}
	r_s(\boldsymbol{x}) &= \sum_{b=0}^{n_g-1}\!\left(1-\exp\!\Big(-\frac{\sum_{i=0}^{n_v-1} x_{s b + i}^2}{2\sigma^2}\Big)\right)
	\label{osl0}
	= \sum_{b=0}^{n_g-1}\!\left(1-\prod_{i=0}^{n_v-1} e_{s b + i}\right), \quad
	e_j \;=\; \exp\!\Big(-\frac{x_j^2}{2\sigma^2}\Big).
\end{align}

where $n_g$ is the number of groups. This function, unlike the SL0 function in Definition \ref{def_sl0}, considers all elements in each group to be a unit of sparsity regularization. 

Assume $\sigma \to 0$, then the exponential function in \eqref{osl0} is zero unless all elements in the group are zero (non-active). Thus, as $\sigma \to 0$, OSL0 considers the maximum loss for a group with even one weakly active (just non-zero) element. This approach aims to select the fewest groups according to the objective function to minimize the amount of loss. In addition, in the overlap mode, particularly for $s \leq \lfloor \frac{n_v}{2} \rfloor$, this selection is localized because selecting among overlapping groups activates fewer groups compared to separate selection. 
For $s = n_v \geq 2$, the function approximates group sparsity without overlapping, since there is no overlap among the groups. For $s = n_v = 1$, it approximates element-wise sparsity, similar to SL0 in Definition \ref{def_sl0}.

Based on \eqref{osl0}, the gradient for OSL0 is:
\begin{align}
	\label{osl0-grad}
	\frac{\partial r_s(\boldsymbol{x})}{\partial x_j}
	&= \frac{x_j}{\sigma^2}\,E_{T_j}, \quad
	E_{T_j} \;=\; \sum_{b\in B_j}\prod_{i=0}^{n_v-1} e_{s b + i},
\end{align}
where $B_j$ is the set of all groups that include $x_j$. When $\sigma \to \infty$ then $\forall\, x_{s b + i} \ll \sigma$ and therefore $e_{s b + i}=\exp\!\big(-x_{s b + i}^2/(2\sigma^2)\big)\to 1$. Hence, from \eqref{osl0-grad} the gradient at $x_j$ is given by $x_j\cdot n_{B_j}/\sigma^2$, where $n_{B_j}$ denotes the number of groups that $x_j$ belongs to. Choosing the gradient step size proportional to $\sigma^2$ in a first-order optimization, OSL0 approximates $\|\cdot\|_2^2$ while each input element is effectively weighted by the square root of the number of groups it belongs to. In the non-overlapping case ($s=n_v$), each element belongs to only one group, so OSL0 is an exact approximation of $\|\boldsymbol{x}\|_2^2$, analogous to SL0 in Definition~\ref{def_sl0}.

To analyze the convergence behavior, we first investigate the convexity of the OSL0 function. The following theorem provides a sufficient condition for convexity:
\begin{theorem}
	\label{thrm-lip}
	Let $\boldsymbol{x}$ be grouped as in \eqref{ogrouping}. Then, the corresponding OSL0 function $r_s(\cdot)$ is convex if  
	$\sigma \geq x_m \sqrt{2n_v - 1}$,
	where $x_m$ is the maximum absolute value of the element in $\boldsymbol{x}$.
	\label{thrm-convexity}
\end{theorem}
\begin{proof}
	To check the convexity condition of $r_s(\cdot)$ in \eqref{osl0}, with $L_\sigma = 1/\sigma^2$, the elements of the Hessian matrix ($\boldsymbol{H}$) are:
	\begin{align}
		\label{eh1}
		&\frac{\partial^2  r_s(\boldsymbol{x})}{\partial x_j^2} = L_\sigma(1-L_\sigma x_j^2) E_{T_j},
		\\
		\label{eh2}
		&\frac{\partial^2  r_s(\boldsymbol{x})}{\partial x_j\partial x_k} = -L_\sigma^2 x_j x_k E_{T_{j,k}},
		\quad
		E_{T_{j,k}} = \sum_{b \in B_{j,k}}\prod_{i=0}^{n_v-1} e_{s b + i},
	\end{align}
	where $B_{j,k}$ and $n_{B_{j,k}}$ are the set of all groups including both elements $x_j$ and $x_k$ and the number of such groups, respectively. According to the second condition of convexity ($\boldsymbol{v}^\top \boldsymbol{H} \,\boldsymbol{v} \geq 0$ for every $\boldsymbol{v}$), the equation below can be derived.
	\begin{align}
		\nonumber
		\boldsymbol{v}^\top \boldsymbol{H}\,\boldsymbol{v} & = L_\sigma \sum_{j=0}^{N-1} v_j^2  (1-L_\sigma x_j^2)  E_{T_j}
		- 
		 \; L_\sigma^2 \sum_{j=0}^{N-1} v_j x_j\Big( \sum_{\substack{k=j-n_v+1\\ k \neq j}}^{\,j+n_v-1} v_k x_k \,E_{T_{j,k}}\Big) 
		\\
		\label{1} 
		&= L_\sigma \Big(
		\sum_{j=0}^{N-1} v_j^2 E_{T_j} - L_\sigma 
		\sum_{j=0}^{N-1} v_j x_j \sum_{k=j-n_v+1}^{j+n_v-1} v_k x_k E_{T_{j,k}}\Big).	
	\end{align}
	In this equation, we consider the maximum neighboring for each element based on \eqref{ogrouping} ($n_v-1$ elements before and after each one), while some of the paired elements (e.g.\ $x_j$ and $x_k$) are not even in a mutual group. So, they cannot be included in the second term of \eqref{1}. However, as we will consider the absolute value of variables in the second term to solve the problem in the worst case, this issue results in only a loose approximation for certain strides. Since $B_{j,k} \in \{B_j, B_k\}$, it always holds that
	\[
	0 \leq E_{T_{j,k}} \leq E_{T_j},\; E_{T_k}.
	\]
	Using this fact, $L_\sigma > 0$, and Lemma~\ref{lemm_binomial}, we obtain:
	\begin{align}
		\label{3}
		 	A &:= \sum_{j=0}^{N-1} \sum_{k=j-n_v+1}^{j+n_v-1} |v_j|\, |v_k|\, E_{T_{j,k}}
		\leq 0.5 \sum_{j=0}^{N-1} \sum_{k=j-n_v+1}^{j+n_v-1} (v_j^2 + v_k^2) E_{T_{j,k}}
		\; \leq \; B,
		\\
		& B:= 
		0.5 \Big(\sum_{j=0}^{N-1}   \sum_{k=j-n_v+1}^{j+n_v-1} v_j^2 E_{T_j} +\sum_{j=0}^{N-1} \sum_{k=j-n_v+1}^{j+n_v-1} v_k^2 E_{T_k} \Big).
	\end{align}
	In the above, because $k$ depends on $j$ ($k \in [j-n_v+1,\ldots,j+n_v-1]$), we have:
	\begin{align}
		\nonumber
		& \sum_{j=0}^{N-1} \sum_{k=j-n_v+1}^{j+n_v-1} v_k^2 E_{T_k} \leq (2n_v-1) \sum_{j=0}^{N-1} v_j^2 E_{T_j}  \Rightarrow
		\label{6}
		A \leq B \leq (2n_v-1)\sum_{j=0}^{N-1} v_j^2 E_{T_j}. \\
	\end{align}
	Therefore, with $x_m = \max(|x_0|, \ldots, |x_{N-1}|)$, we reach:
	\begin{align}
		\frac{\boldsymbol{v}^\top \boldsymbol{H}\,\boldsymbol{v}}{L_\sigma} & \geq
		\sum_{j=0}^{N-1} v_j^2 E_{T_j} - L_\sigma x_m^2  (2n_v-1)\sum_{j=0}^{N-1} v_j^2 E_{T_j} 
		=
		\sum_{j=0}^{N-1} v_j^2 E_{T_j} \big(1 - L_\sigma x_m^2  (2n_v-1)\big).
	\end{align}
	Finally, according to the second condition of convexity, OSL0 is convex if:
	\begin{align}
		\label{9}
		& 1 - \dfrac{1}{\sigma^2} x_m^2 (2n_v-1) \geq 0 \quad
		\Rightarrow \quad
		\sigma \geq x_m \sqrt{2n_v-1} .
	\end{align}
\end{proof}

This theorem establishes that OSL0 is convex for sufficiently large $ \sigma $, ensuring convergence to a global minimum.  

However, in practical applications of sparsity regularization, $ \sigma $ must be small to enforce sparsity, which results in the OSL0 function being non-convex when 
$\sigma \to 0$.
In this non-convex area, convexity-based guarantees are not held. To address this issue, we analyze the Lipschitz condition of OSL0, which allows us to control the gradient step size and ensure convergence to at least a local minimum.
The following theorem establishes the Lipschitz constant for OSL0:

\begin{theorem}
	\label{thrm-conv}
	The Lipschitz constant $ L $ for OSL0, when the input follows the grouping in \eqref{ogrouping}, is given by $L \geq  (1 + 2(2n_v - 3)\exp(-1) ) \lceil \frac{n_v}{s} \rceil /\sigma^2$.
	\label{thrm-Lipschitz}
\end{theorem}

\begin{proof}
	Assume $\boldsymbol{M} = L\mathbf{I}-\boldsymbol{H}$ where $\boldsymbol{H}$ and $\mathbf{I}$ are respectively the Hessian matrix of $ r_s(\boldsymbol{x})$ in \eqref{osl0} and the identity matrix with appropriate dimensions. $L$ is the Lipschitz constant if $\boldsymbol{M}$ is a positive semi-definite matrix ($\boldsymbol{v}^\top \boldsymbol{M} \boldsymbol{v} \geq 0$ for all $\boldsymbol{v}$). Given that \eqref{eh1} and \eqref{eh2} are elements of $\boldsymbol{H}$, we have:
	
	\begin{align}
		\label{LC}
		\boldsymbol{v}^\top \boldsymbol{M} \boldsymbol{v} & =  \sum_{j=0}^{N-1} v_j^2\big[L - L_\sigma (1-L_\sigma x_j^2)  E_{T_j}\big] + L_\sigma^2 C,
		\\
		\label{10}
		C & 
		= 
		\sum_{j=0}^{N-1} v_j x_j \sum_{\substack{k=j-n_v+1\\ k \neq j}}^{j+n_v-1} v_k x_k \,E_{T_{j,k}}. 
	\end{align}
	In \eqref{10}, $C$ exists for $n_v \geq 2$, and it can be positive or negative. Accordingly, we assume the worst case as negative: 
	\[
	C \geq - \sum_{j=0}^{N-1} |v_j x_j| \sum_{\substack{k=j-n_v+1\\ k \neq j}}^{j+n_v-1} |v_k x_k| \,E_{T_{j,k}}.
	\]
	Then, similar to \eqref{3} by using Lemma \ref{lemm_binomial}:
	\begin{align}
		C &\geq - 
		\frac{1}{2} \sum_{j=0}^{N-1} \sum_{\substack{k=j-n_v+1\\ k \neq j}}^{j+n_v-1} \big(x_j^2 v_j^2 + x_k^2 v_k^2\big) E_{T_{j,k}} 
		\label{11}
		\geq  - (2n_v-2)\sum_{j=0}^{N-1}x_j^2 v_j^2 E_{T_j} .
	\end{align}
	Therefore,
	\begin{align}
		\boldsymbol{v}^\top \boldsymbol{M} \boldsymbol{v} 
		&\geq L \sum_{j=0}^{N-1} v_j^2 - L_\sigma
		\sum_{j=0}^{N-1} v_j^2 E_{T_j} 
		-  2L_\sigma (2n_v-3) \sum_{j=0}^{N-1} v_j^2 \frac{1}{2\sigma^2} x_j^2 E_{T_j}.
	\end{align}
	Assume $\sigma \neq 0$. Then 
	\[
	0\leq e_{s b + i} = \exp\!\Big(-\frac{x^2_{s b + i}}{2\sigma^2}\Big)\leq 1, \quad \forall\, x_{s b + i} \in \mathbb{R}.
	\]
	Thereby, for $b \in B_j$, $0\leq \prod_{i=0}^{n_v-1}e_{s b + i} \leq e_{j} \leq 1$, and then 
	\begin{equation}
		\label{12-1}
		0\leq E_{T_j}\leq n_{B_j} e_j \leq n_{B_m} e_j \leq n_{B_m},
	\end{equation}
	where $n_{B_m} = \big\lceil \dfrac{n_v}{s}\big\rceil$ is the maximum of $n_{B_i}$ for all $i$.
	Based on \eqref{12-1} and Lemma \ref{lemm_bound_exp}, we conclude: 
	\[
	\frac{x_j^2 E_{T_j}}{2\sigma^2} \leq e_j n_{B_m} \frac{x_j^2}{2\sigma^2} \leq n_{B_m} e^{-1}.
	\]
	Hence,
	\begin{align}
		\nonumber
        \boldsymbol{v}^\top \boldsymbol{M} \boldsymbol{v} 
		&
		\geq L \sum_{j=0}^{N-1} v_j^2 - L_\sigma
		n_{B_m}\sum_{j=0}^{N-1} v_j^2  -  \quad 2L_\sigma (2n_v-3)e^{-1} n_{B_m}\sum_{j=0}^{N-1} v_j^2 \\
		& =
		\Big(L - L_\sigma n_{B_m} \big[1 + 2(2n_v-3)e^{-1}\big]\Big) \sum_{j=0}^{N-1} v_j^2 .
	\end{align}
	
	Finally, to meet the Lipschitz condition: 
	\begin{equation}
		\label{lco}
		L \geq  \frac{1}{\sigma^2 } \big(1 + 2(2n_v-3)e^{-1} \big) \bigg\lceil \frac{n_v}{s}\bigg\rceil.
	\end{equation}
\end{proof}

This result enables us to determine a suitable gradient step size $\eta \leq 1/L$,
which ensures the convergence of the optimization process, even in the non-convex setting.

\begin{theorem}
	\label{thrm-whole}
	Consider the composite objective in \eqref{eq-main-attack}:
	\[
	f(\boldsymbol{\Delta}) = \mathrm{CE}(\boldsymbol{y}_{\delta}, \boldsymbol{y}_t) + \lambda_s \cdot r_s(\mathcal{C}(\boldsymbol{\Delta};\mathcal{R}))  
	+ \lambda_\infty \cdot r_\infty(\mathcal{V}(\boldsymbol{\Delta})),
	\]
	and let $\{\boldsymbol{\Delta}_k\}$ be the iterates of our first-order optimization method.
	Suppose the following conditions hold:
	\begin{enumerate}
		\item[\textnormal{(i)}] KL property: \(f\) is proper, lower semi-continuous, and satisfies the Kurdyka-\L{}ojasiewicz property on \(\operatorname{dom}\,\partial f\).
		\item[\textnormal{(iv)}] Boundedness: The sequence \(\{\boldsymbol{\Delta}_k\}\) has at least one finite limit point.
	\end{enumerate}
	Then the entire sequence \(\{\boldsymbol{\Delta}_k\}\) converges to a critical point of \(f\) \cite{xu2017globally}.
\end{theorem}

\begin{proof}
	Since \(r_\infty(\cdot)\) is convex \cite{heshmati2022designing} and \(r_s(\cdot)\) satisfies Lipschitz or convexity conditions depending on \(\sigma\) (as in Table~\ref{tab:osl0-init}), one can verify the sufficient decrease and relative error conditions as in \cite{amini2024fast}. Then, by the KL property of \(f\), the convergence of the entire sequence to a stationary point follows from standard results in nonconvex optimization \cite{xu2017globally}.
\end{proof}

Our approach to using the OSL0 function begins with convex initialization, where we set $\sigma$ to satisfy the convexity condition in Table \ref{tab:osl0-init}, ensuring a well-behaved optimization landscape. As the function transitions out of convexity due to decreasing $\sigma$, we employ a Lipschitz-guided step size, using the Lipschitz constant from Table \ref{tab:osl0-init} to determine an appropriate gradient step size that ensures stable convergence of Problem \ref{eq-main-attack} to a stationary point. This strategy ensures stable convergence even when the OSL0 function becomes non-convex as sparsity is gradually enforced.

\subsection{Non-Overlapping SL0}
The Non-Overlapping SL0 (NOSL0) is a special case of OSL0 where $s = n_v$. 
The original SL0 is a specific version of NOSL0 where $s = n_v = 1$. 
In this case, the bounds for the Lipschitz constant will be $1/\sigma^2$ because $C$ in \eqref{LC} is 0, which is tighter than the bound in \eqref{lco}. 
For NOSL0 where $s = n_v \geq 2$, the vector $\mathbf{x}$ is grouped as follows:
\begin{equation}
	\label{nogrouping}
	\mathbf{x} =
	\left[
	\underbrace{x_0,\,\ldots,\,x_{n_v-1}}_{\text{Group-1}},
	\underbrace{x_{n_v},\,\ldots,\,x_{2n_v-1}}_{\text{Group-2}},
	\ldots
	\right]^{\!T}.
\end{equation}
We can use $n_v$ for the maximum number of elements being neighbors instead of $2n_v-1$ for OSL0, and we can replace this maximum in \eqref{6} and \eqref{11}. Therefore, the convexity condition will be $\sigma \geq x_m \sqrt{n_v}$, and the Lipschitz constant can be bounded by $L \geq  (1 + 2(n_v-2)\exp(-1) )/\sigma^2$.

The results regarding convexity and the Lipschitz condition for all OSL0 settings are summarized in Table \ref{tab:osl0-init}. To use this table, $s$ and $n_v$, specifying the mode of the OSL0 function, should be defined. Furthermore, $\sigma_m$, which is the maximum absolute change allowed for all input elements in optimization, should be approximated if there are no limits on the changes.

	\begin{table}[t!!!]
	\caption{Convexity and the Lipschitz condition of the extensions of OSL0 due to the input vector ($\delta$) like \eqref{ogrouping}.
	}
	\renewcommand{\arraystretch}{1.5}
	\label{tab:osl0-init}
	\centering
	\resizebox{0.7\textwidth}{!}{%
		\begin{tabular}{c|cc}
			\hline
			Function Name &
			Convexity &
			Lipschitz Condition ($\eta \leq \frac{1}{L}$) \\ \hline\hline
			
			OSL0 ($n_v > s$)   &
			$\sigma \geq \delta_m \sqrt{2n_v-1}$      & 
			$ \quad L \geq (1 + 2(2n_v-3)\exp(-1) ) \lceil \frac{n_v}{s} \rceil /\sigma^2$
			\\
			NOSL0 (OSL0 where $n_v = s>1$)    &     $\sigma \geq \delta_m \sqrt{n_v}$     &    $L \geq   (1 + 2(n_v-2)\exp(-1) )/\sigma^2
			$		
			\\
			SL0 (OSL0 where $n_v = s = 1$) &
			$\sigma \geq \delta_m$ & 
			$L \geq  1/\sigma^2 $        
			\\
			\hline
		\end{tabular}
	}
\end{table}

\subsection{Adjusting ATOS with OSL0}
In this section, we represent all kinds of sparsity modes of ATOS based on the framework in \eqref{eq-main-attack} and OSL0 regularization introduced in \eqref{osl0}. To design element-wise adversarial attacks, OSL0 can be used with $n_v = s = 1$, meaning that all input elements are individually important. In other words, $\mathcal{R}$ is the rule of normal vectorizing in this case. Hence, OSL0 is equivalent to SL0 as defined in Definition \ref{def_sl0}. SL0 regularization has shown to be a strong and fast term for gradually curbing the $\ell_0$ norm in different applications such as compressed sensing \cite{mohimani2008fast, Sadeghi2016}. Following this case, pixel-wise and group-wise attacks will be proposed.

\subsubsection{Pixel-Wise Attack}
Since images are naturally two-dimensional, altering each element of a pixel with $c$ channels (e.g., Blue, Red, and Green) makes the entire pixel appear perturbed to human vision, regardless of how many of its channels are manipulated. Therefore, to minimize visual perturbation, it is important to directly control the number of pixels perturbed during an attack rather than the number of channel elements. To design our pixel-wise attack, the sparsity regularization in \eqref{eq-main-attack}, OSL0, must be without overlapping ($n_v = s$), which results in NOSL0. In this attack, the categorizing rule ($\mathcal{R}$) considers each pixel as a group. NOSL0, like the main SL0, requires initialization of its parameter $\sigma$ before optimization. Although \cite{mohimani2008fast} recommends starting with a large value of $\sigma$ to approximate the $\ell_2$ norm initially, we propose to initialize $\sigma$ with a smaller value as specified in Table \ref{tab:osl0-init}, to maintain convexity while also increasing the convergence rate. In this table, the maximum range of perturbation ($\delta_m$) must be defined.

\subsubsection{Group-Wise Attack}
In this section, our structured group-wise adversarial attack is detailed. First, since the attack requires an overlapping categorizing rule, we suggest using a window of size $n \times n$ and depth of $c$, the number of image channels, similar to \cite{xu2018structured}. The window is then moved across the image with stride $s$ in horizontal and vertical directions, grouping all elements within each window. In this $\mathcal{R}$, there is no overlap among groups if $s = n$, and some input elements are excluded when $s > n$, making such strides unsuitable for designing a structured attack. In contrast, there is an overlap among neighboring groups for $s < n$. In addition, for $s \le \lfloor \frac{n}{2} \rfloor$, all elements in a group are included in at least one other group, establishing indirect connections between all elements of the image through links among all groups.

 In general, as $s$ decreases, the intensity of the overlap increases, reaching its maximum when $s = 1$. The stride represents the highest resolution for perturbation in a group-wise attack, as the remaining $n - s$ vertical and horizontal elements are placed in different groups for $s \le \lfloor \frac{n}{2} \rfloor$. When these elements become active, the corresponding groups are also activated. Consequently, the OSL0 function computes the loss for them as well. 
 
 In addition, the larger $n$ is, the more elements are neighbors within a group, and the more groups are directly connected to each other. This is useful for creating direct connections between more elements of the image. However, a smaller $s$ and larger $n$ increase the number of groups and the number of neighbors in each group, respectively, leading to higher computational complexity with the order of $O(n^2/s^2)$. With this categorization, a group-wise attack can be formulated based on \eqref{eq-main-attack}.
 
  In the group-wise attack, $\mathcal{R}$ must be applied to all three dimensions ($c \times w \times h$) with overlap, meaning that the resulting grouped pattern does not form a regular vectorized input pattern as in \eqref{ogrouping}. To address this, we modify the conditions for OSL0 in Table \ref{tab:osl0-init} by replacing $2n_v-1$ and $\big\lceil \frac{n_v}{s} \big\rceil$ with $n_n$ and $n_{B_m}$, respectively, which represent the maximum number of neighbors for elements and the maximum number of groups an element can belong to. 
 Despite this modification, we recommend using Table \ref{tab:osl0-init} with $n_v = 3n$, especially for initializing the first $\sigma$, since increasing $\sigma$ leads to a slower convergence rate. The complete algorithm for our three ATOS attacks is presented in Algorithm \ref{alg:main}. In this algorithm, we multiply the OSL0 function by $\sigma_k^2$ in the step \ref{grad-step} because, according to Table \ref{tab:osl0-init}, the gradient step size of this function for all sparsity modes is related to $\sigma_k^2$, and the remaining parts are considered as $\lambda_s$. It is important to note that Step \ref{quantization-step} is crucial for generating feasible adversarial examples and making OSL0 useful for sparsity. This is because the changes in each iteration before quantization are continuous and can be close to zero, but not exactly zero in some cases. Therefore, this step guarantees that the attacks maintain pure sparsity.

The next section discusses the implementation of our ATOS attacks using this algorithm and provides a comparison with other frameworks.
\begin{algorithm}[t!]
	\caption{ATOS}\label{alg:main}
	
	\SetKwInput{KwObjective}{Objective}
	\SetKwInput{KwInput}{Input}
	\SetKwInput{KwInit}{Initialization}
	
	\KwObjective{Element-wise, Pixel-wise or Group-wise Adversarial Perturbation, $\boldsymbol{\Delta}$.}
	
	\KwInput{
		Main input $\boldsymbol{X}$, Adversarial target $\boldsymbol{y}_t$, 
		Number of steps $N_s$ and iterations $N_i$ per step, 
		Total gradient step size $\mu$, 
		Relaxations ($\lambda$, $\lambda_s$, $\lambda_\infty$) and scale $s_\lambda$ for $\lambda$, 
		Hyperparameters ($\sigma_0$, $s_\sigma$, $p$), 
		Rule of categorizing $\mathcal{R}$ to design element-wise, pixel-wise or group-wise perturbation.
	}
	
	\KwInit{$\boldsymbol{\Delta}_0 = 0$, $k = 0$.}
	
	\While{$k \le N_s-1$}{\label{while}
		\For{$i=0, \dots, N_i-1$}{
			Define $\boldsymbol{y}_{\sigma_i}$ by using $\boldsymbol{\Delta}_i + \boldsymbol{X}$ as input to the model. 
			
			$
			\boldsymbol{\Delta}_{i+1} = 
			\boldsymbol{\Delta}_i - \mu \cdot 
			\nabla_{\boldsymbol{\Delta}_i} 
			\Bigg\langle 
			\lambda \cdot \text{CE}(\boldsymbol{y}_{\sigma_i}, \boldsymbol{y}_t) + 
			\lambda_s \sigma_k^2 \cdot r_s(\mathcal{C}(\boldsymbol{\Delta}_i; \mathcal{R})) \Big|_{\sigma = \sigma_k} + 
			\lambda_\infty \cdot r_\infty(\mathcal{V}(\boldsymbol{\Delta}_i)) \Big|_P
			\Bigg\rangle
			$ \label{grad-step}
			
			Clip $\boldsymbol{\Delta}_{i+1}$ so that $\boldsymbol{\Delta}_{i+1} + \boldsymbol{X} \in [0,1]$. 
		}
		$\boldsymbol{\Delta}_0 = \boldsymbol{\Delta}_{N_i}$; \\
		$\sigma_{k+1} = \sigma_k \times s_\sigma$; \\
		$k = k+1$;
	}
	
	Quantize $\boldsymbol{\Delta} = \boldsymbol{\Delta}_{N_i}$ into the real image range $[0,1]$\label{quantization-step}. 
	
	Define $\boldsymbol{y}_\sigma$ by using $\boldsymbol{\Delta} + \boldsymbol{X}$ as input to the model\label{modifying-step}. 
	
	\uIf{$\boldsymbol{y}_\sigma \neq \boldsymbol{y}_t$}{
		$\boldsymbol{\Delta}_0 = 0$, $k = 0$, $\lambda = \lambda \times s_\lambda$; \\
		Go to Step \ref{while};
	}
\end{algorithm}


\section{Simulation Results} \label{sr}	
Thus far, our element-wise, group-wise, and pixel-wise adversarial attacks have been proposed by ATOS using accurate approximations of sparsity patterns and the gradient of the $\ell_\infty$ norm. 

\subsection{Experimental Setup}
In this section, we implement our ATOS attacks based on Algorithm \ref{alg:main} and evaluate them comparatively. Specifically, GSA-SCAD \cite{amini2024fast} and PGD-$\ell_0$ \cite{saich19} are implemented for element-wise and pixel-wise modes, as they represent a new attack followed by ATOS and one of the strongest available pixel-wise attacks, respectively. In the group-wise mode, we compare our ATOS framework with other frameworks, including StrAttack \cite{xu2018structured}, Homotopy \cite{zhu2021sparse}, FWnucl \cite{kazemi2023minimally}, and GSE \cite{ICLR2025_398b00a0}.

We use five different models for the CIFAR-10 \cite{cifar10} and ImageNet \cite{deng2009imagenet, russakovsky2015imagenet} datasets. The first one, referred to as CNN, consists of simple yet effective CNN blocks and achieves an accuracy of $79.51\%$ on CIFAR-10 \cite{saafw20}. For ImageNet, we use ViT-s16 \cite{touvron2022three,dosovitskiy2020image}, an attention-based architecture with $79.7\%$ top-1 accuracy, and ResNet152 \cite{wightman2021ResNet}, which achieves $82.74\%$ top-1 accuracy. The final two models are the MeanSparse ConvNeXt-L \cite{amini2024meansparse} and the Swin-L Transformer \cite{liu2024comprehensive}, which are highly ranked robust models on RobustBench \cite{NEURIPS_DATASETS_BENCHMARKS2021_a3c65c29}, achieving AutoAttack accuracies of $58.22\%$ and $59.56\%$, respectively. In addition, we use all correctly classified samples from the test set of CIFAR-10. For ImageNet, we randomly select 1,000 correctly classified samples from the validation set. 

The implementation of our ATOS attacks is based on Algorithm \ref{alg:main} and the code is available at: \url{https://github.com/alireza-heshmati/ATOS}.

ATOS attacks can be either targeted or untargeted. Targeted attacks can be further divided into three main categories: best, worst, and average \cite{carlini2017towards}. The best case aims to deceive the classifier into assigning the non-maximal class with the highest confidence, while the worst case seeks to assign the class with the lowest confidence. The average case involves fooling the classifier into assigning a random non-maximal class. In this paper, we focus on the worst case, which represents the most challenging scenario, and also consider the untargeted case as the least challenging.

Regarding Algorithm \ref{alg:main}, $N_i$ is set to $200$, and $p$ is set to $10^4$ to control the intensity of perturbations. 
Based on \cite{mohimani2008fast}, we select 0.5 for $s_\sigma$ to enforce strong sparsity for all attacks, except those that use the LSEAp function on the ImageNet dataset to reduce the intensity of the perturbation. This is because the image size is larger than that of the CIFAR-10 dataset, which helps the LSEAp function be applied more effectively.
$\mathcal{R}$ for our group-wise attacks is based on a sliding window with $(n = 4, s = 1)$ for the CIFAR-10 dataset and $(n = 16, s = 2)$ for the ImageNet dataset. 
This selection of $n$ and $s$ is based on maximizing neighborhood coverage while maintaining suitable computational complexity, which is of the order $\mathcal{O}(n^2 / s^2)$, as discussed in the previous section.
To determine $ \sigma_0 $ based on Table~\ref{tab:osl0-init}, for attacks that do not control the perturbation intensity and whose aim is solely to induce sparse modes, the maximum absolute perturbation (denoted $ \delta_m $ in the table) is set to 1 since it is the maximum allowable change for each element of the images normalized between 0 and 1. In cases where it is necessary to curb the perturbation intensity, however, $ \delta_m $ must be loosely approximated, as demonstrated in each experiment alongside the other parameters.

Although relaxation parameters ($\lambda$, $\lambda_s$, and $\lambda_\infty$) and the gradient step ($\mu$) are selected through a search over a hyper-parameter grid, it is better to first select $\mu$ and $\lambda$ in a way that results in adversarial perturbation with the lowest $\ell_\infty$ norm. Then, the other parameters should be determined, and even the previous ones should be re-selected if necessary. To tune the hyperparameters in Algorithm \ref{alg:main}, we provide the step \ref{modifying-step}, in which the constant of the objective function is multiplied by $s_\lambda$.
The other frameworks are implemented with a careful search over a hyper-parameter grid to find the best values. All algorithms in the paper start with a batch of images, which terminate if all images in the batch are successfully attacked. Otherwise, for the unsuccessful ones,  hyper-parameters of frameworks become relaxed to ultimately attack them.

To accurately evaluate the group-wise methods, we use two additional metrics: the interpretability score (IS) and $\ell_{2,0}$. The IS \cite{xu2018structured} is a metric that measures the overlap between the adversarial perturbation ($\Delta$) and the adversarial saliency map (ASM), which offers a sensitivity analysis of the impact of pixel-level perturbations on classification \cite{papernot2016limitations}. Accordingly, 
$IS(\Delta)$ is defined as $\|B_{ASM} \odot \Delta\|_2/\|\Delta\|_2$, where $B_{ASM}$ is a binary mask derived from the ASM, with its elements set to one if the corresponding elements of the ASM exceed a threshold $v$, and zero otherwise. Since the ASM can measure the effect of each pixel of the perturbation on the target, when $IS(\Delta) \to 1$, more perturbation locations match the regions of the image that are more important for inference. Thus, a higher IS indicates a more explainable attack.
The $\ell_{2,0}$ metric counts the number of groups in the adversarial perturbation, containing at least one non-zero element. These groups, which are fully overlapped, are achieved by moving a window with a size of $8\times8$, depth 3, and stride 1 for both the CIFAR-10 and ImageNet datasets.

\subsection{Results}
In this section, we investigate the results of ATOS for its element-wise, pixel-wise, and group-wise perturbations, referred to as ATOS-$\ell_0$, ATOS-PW, and ATOS-GW, respectively. For ATOS attacks that control $\ell_\infty$ norm, the suffix "-$\ell_\infty$" is appended.

\begin{table}
    \centering
	\caption{Comparison of the proposed element-wise and pixel-wise attacks (ATOS-$\ell_0$ and ATOS-PW, respectively) to related methods. Common parameters are: $N_s: 10; s_\lambda: 2$.}
	\label{tab:sparse}
	\resizebox{0.65\textwidth}{!}{%
		\begin{tabular}{ccccccccccc}
			\hline 
			Network & \multirow{2}{*}{Methods} & 
			\multicolumn{4}{c}{Untargeted Case}& \multicolumn{5}{c}{Worst Case} \\ \cline{3-6} 
			\cline{8-11} 
			(Dataset) &  & ASR & NPP  & $\ell_0$  & $T$  && ASR  & NPP  & $\ell_0$  & $T$  \\ \hline
			\multirow{3}{*}{CNN}
			& PGD-$l_0$
			& $\boldsymbol{100.0}$ & $7$  & $21$  & $\boldsymbol{0.024}$ & 
			& $\boldsymbol{100.0}$ & $22$  & $65$  & $0.033$   \\   
			\multirow{3}{*}{(CIFAR-10)}
			& GSA-SCAD
			& $\boldsymbol{100.0}$ & $15$  & $21$  & $0.034$ &  
			& $\boldsymbol{100.0}$ & $38$  & $51$  &  $0.034$   \\ 
			
			& ATOS-$\ell_0$    
			& $\boldsymbol{100.0}$ & $10$  & $\boldsymbol{14}$   & $0.033$  & 
			& $\boldsymbol{100.0}$ & $25$  & $\boldsymbol{35}$   & $\boldsymbol{0.032}$   \\ 
			
			& ATOS-PW  
			& $\boldsymbol{100.0}$ & $\boldsymbol{6}$  & $18$ 
			& $0.032$  & 
			& $\boldsymbol{100.0}$ & $\boldsymbol{16}$  & $47$ & $\boldsymbol{0.032}$   \\ 
			\cline{2-11}
			
			\multicolumn{11}{c}{\footnotesize ATOS-$\ell_0$($\mu: 1.2; c: 0.0125;   \lambda_s: 1$), ATOS-PW($\mu: 1.4; c: 0.0125;   \lambda_s: 1$)}\\
			
			\hline
			\multirow{3}{*}{ResNet152}
			
			& PGD-$l_0$
			& $\boldsymbol{100.0}$ & $393$  & $1179$  & $\boldsymbol{3.16}$ & 
			& $\boldsymbol{100.0}$ & $2496$  & $7478$  & $5.60$   \\   
			\multirow{3}{*}{(ImageNet)}
			& GSA-SCAD
			& $\boldsymbol{100.0}$ & $503$  & $548$  & $17.6$ &  
			& $99.9$ & $1338$  & $1532$  & $17.7$   \\ 
			
			& ATOS-$\ell_0$    
			& $\boldsymbol{100.0}$ & $131$  & $\boldsymbol{141}$  & $5.54$  & 
			& $\boldsymbol{100.0}$ & $596$  & $\boldsymbol{691}$  & $\boldsymbol{5.55}$   \\ 
			
			& ATOS-PW  
			& $\boldsymbol{100.0}$ & $\boldsymbol{98}$  & $256$  & $5.53$  & 
			& $\boldsymbol{100.0}$ & $\boldsymbol{396}$  & $1160$  &  $5.60$    \\ 
			\cline{2-11}
			
			\multicolumn{11}{c}{\footnotesize ATOS-$\ell_0$($\mu: 1.2; c: 0.025;   \lambda_s: 0.8,0.6$), ATOS-PW($\mu: 1.2; c: 0.025;   \lambda_s: 0.8,0.6$)}\\
			\hline
			\multirow{3}{*}{ViT-s16}
			
			& PGD-$l_0$
			& $\boldsymbol{100.0}$ & $170$  & $509$  & $\boldsymbol{1.25}$ & 
			& $\boldsymbol{100.0}$ & $2966$  & $8883$  & $3.34$   \\   
			\multirow{3}{*}{(ImageNet)}
			& GSA-SCAD
			& $99.9$ & $\boldsymbol{142}$  & $\boldsymbol{165}$  & $5.23$ & 
			& $\boldsymbol{100.0}$ & $ 520$  & $\boldsymbol{631}$  & $4.37$   \\ 
			
			& ATOS-$\ell_0$    
			& $\boldsymbol{100.0}$ & $283$  & $371$  & $2.48$  & 
			& $\boldsymbol{100.0}$ & $705$  & $899$  & $\boldsymbol{2.44}$   \\ 
			
			& ATOS-PW  
			& $\boldsymbol{100.0}$ & $211$  & $605$  & $2.53$  & 
			& $\boldsymbol{100.0}$ & $\boldsymbol{501}$  & $1475$  & $2.45$   \\ 
			\cline{2-11}
			\multicolumn{11}{c}{\footnotesize ATOS-$\ell_0$($\mu: 1.2; c: 0.025;   \lambda_s: 0.2$), ATOS-PW($\mu: 1.2; c: 0.025;   \lambda_s: 0.2$)}\\
			\hline
		\end{tabular}
	}
\end{table}
\begin{table}
	\caption{Comparison of the proposed imperceptible pixel-wise attack (ATOS-PW-$\ell_\infty$)  to competing methods. Common parameters are: $\mu: 0.2.$}
	\label{tab:invisible}
	\resizebox{0.8\textwidth}{!}{
		\begin{tabular}{ccccccccccc}
			\hline 
			Network & \multirow{2}{*}{Methods} & 
			\multicolumn{4}{c}{Untargeted Case}& \multicolumn{5}{c}{Worst Case} \\ \cline{3-6} 
			\cline{8-11} 
			(Dataset)&  & ASR & NPP  & $\ell_\infty$   & $T$  && ASR  & NPP  & $\ell_\infty$   & $T$  \\ \hline
			\multirow{3}{*}{CNN}
			
			& PGD-$(l_0$,$l_\infty)$
			& $95.3$ & $113$  & $0.032\pm0.005$  & $0.08$ & 
			& $\boldsymbol{100.0}$ & $175$  & $0.051\pm0.000$  & $0.03$   \\   
			
			\multirow{3}{*}{(CIFAR-10)}& GSA-SCAD
			& $\boldsymbol{100.0}$ & $76$  & $0.031\pm0.000$  & $0.034$ &  
			& $\boldsymbol{100.0}$ & $133$  & $0.051\pm0.000$  & $0.033$   \\ 
			
			& ATOS-PW-$\ell_\infty$  
			& $\boldsymbol{100.0}$ & $\boldsymbol{48}$  & $0.032\pm 0.016$  & $\boldsymbol{0.016}$  & 
			& $\boldsymbol{100.0}$ & $\boldsymbol{114}$  & $0.049\pm0.01$   & $\boldsymbol{0.016}$   \\ 
			\cline{2-11}
			\multicolumn{11}{c}{\footnotesize ATOS-PW-$\ell_\infty$($N_s: 5 ; c: 0.01,0.02; s_\lambda: 1.5; \lambda_s: 1.2,0.8;\lambda_\infty: 0.15,0.25; \sigma_m: 0.03,0.04$)}\\
			\hline
			\multirow{3}{*}{ResNet152}
			
			& PGD-$(l_0$,$l_\infty)$
			& $\boldsymbol{100.0}$ & $1497$  & $0.02\pm0.001$  & $\boldsymbol{3.62}$ & 
			& $\boldsymbol{100.0}$ & $13049$  & $0.039\pm0.000$  &$\boldsymbol{6.04}$   \\   
			
			\multirow{3}{*}{(ImageNet)} & GSA-SCAD
			& $\boldsymbol{100.0}$ & $851$  & $0.02\pm0.000$   & $13.0$ &  
			& $99.9$ & $\boldsymbol{1709}$  & $0.039\pm0.000$  & $11.9$   \\ 
			
			& ATOS-PW-$\ell_\infty$  
			& $\boldsymbol{100.0}$ & $\boldsymbol{739}$  & $0.019\pm0.009$   & $29.0$  & 
			& $\boldsymbol{100.0}$ & $1896$  & $0.042\pm0.013$   & $37.5$   \\ 
			\cline{2-11}
			
			\multicolumn{11}{c}{ \footnotesize ATOS-PW-$\ell_\infty$($N_s: 10 ; c: 0.01; s_\lambda: 2; \lambda_s: 0.9,1;\lambda_\infty: 0.2; \sigma_m: 0.02,0.04$)}\\
			\hline
			
			\multirow{3}{*}{ViT-s16}
			
			& PGD-$(l_0$,$l_\infty)$
			& $\boldsymbol{100.0}$ & $1368$  & $0.02\pm0.000$  & $\boldsymbol{2.66}$ & 
			& $\boldsymbol{100.0}$ & $16868$  & $0.039\pm0.001$  & $6.33$   \\   
			
			\multirow{3}{*}{(ImageNet)} & GSA-SCAD
			& $99.9$ & $594$  & $0.02\pm0.005$  & $6.41$ &  
			& $\boldsymbol{100.0}$ & $1501$  & $0.039\pm0.006$  & $\boldsymbol{4.09}$   \\ 
			
			& ATOS-PW-$\ell_\infty$  
			& $\boldsymbol{100.0}$ & $\boldsymbol{459}$  & $0.021\pm0.009$  & $8.21$  & 
			& $\boldsymbol{100.0}$ & $\boldsymbol{1373}$  & $0.04\pm0.011$  & $6.05$   \\ 
			\cline{2-11}
			\multicolumn{11}{c}{ \footnotesize ATOS-PW-$\ell_\infty$($N_s: 10 ; c: 0.002,0.005; s_\lambda: 1.5,2; \lambda_s: 0.9,0.6;\lambda_\infty: 0.2,0.3; \sigma_m: 0.02,0.04$)}\\
			\hline
		\end{tabular}
	}
\end{table}

\subsubsection{Pixel-Wise Attack}
Table \ref{tab:sparse} presents a comparison of the proposed pixel-wise attack with other methods. It shows that our pixel-wise adversarial examples consistently have the lowest number of perturbed pixels (NPP), except for cases where the Attack Success Rate (ASR) for GSA-SCAD does not reach $100\%$ for the worst case on ResNet152. In addition, there are clear differences between element-wise attacks (ATOS-$\ell_0$ and GSA-SCAD) and pixel-wise attacks (ATOS-PW and PGD-$\ell_0$) in controlling $\ell_0$ and NPP. Element-wise attacks demonstrate better control over the number of perturbed elements ($\ell_0$), while pixel-wise attacks aim to minimize the number of perturbed pixels (NPP). However, a real attack should control NPP due to the two-dimensional nature of images in reality. Therefore, in the following, we report only the results of our pixel-wise and group-wise attacks.

Table \ref{tab:invisible} demonstrates the results for the case where both sparsity and invisibility are regularized. According to this table, ATOS-PW again achieves the best result in the NPP criterion, with a significant gap between our method and the previous ones. Furthermore, while the mean of the $\ell_\infty$ norm for all attacks is very similar, the standard deviation of the $\ell_\infty$ norm for our attack is higher than that of the others, because we use a loss function (LSEAp) to control this norm, unlike other methods that only use a threshold. However, to ensure fairness in the invisibility of our attack, we employ a threshold ($2\times \sigma_m$) to control the extent of perturbation in each iteration of Algorithm \ref{alg:main}.

\begin{table}
	\caption{Comparison of the proposed group-wise attack (ATOS-GW)  to other methods. Common parameters are: $N_s: 10 ;\mu: 1 ; s_\lambda: 2$.}
	\label{tab:group}
	\resizebox{0.65\textwidth}{!}{
		\begin{tabular}{ccccccccccc}
			\hline 
			
			Network & \multirow{2}{*}{Methods} & 
			\multicolumn{4}{c}{Untargeted Case}& \multicolumn{5}{c}{Worst Case} \\
			\cline{3-6} 
			\cline{8-11} 
			(Dataset) &  & ASR & NPP  &  $\ell_{2,0}$ &  $T$  && ASR  & NPP  &  $\ell_{2,0}$ & $T$  \\ \hline
			\multirow{3}{*}{CNN}
			& FWnucl
			& $78.9$ & $139$  &  $404$ &       $\boldsymbol{0.052}$ &
			& $14.4$ & $182$  &   $462$ &        $\boldsymbol{0.053}$   \\   
			
			\multirow{3}{*}{(CIFAR-10)} & GSE
			& $98.6$ & $29$  &  $150$ &      $1.424$ &
			& $98.7$ & $75$  &  $267$ &        $1.756$   \\   
			
			& ATOS-GW  
			& $\boldsymbol{100.0}$ & $\boldsymbol{15}$   & 
			$\boldsymbol{146}$ &   $0.292$&
			& $\boldsymbol{100.0}$ & $\boldsymbol{52}$  &  
			$\boldsymbol{259}$ &   $0.298$\\ 
			\cline{2-11}
			\multicolumn{11}{c}{ \footnotesize ATOS-GW($ c: 0.0125,0.005; s_\lambda: 2; \lambda_s: 0.05,0.02$ )}\\
			\hline
			\multirow{3}{*}{ResNet152}
			
			& FWnucl
			& $65.5$ & $17663$  & $19647$ &        $6.910$ &
			& $45.3$ & $15944$  &   $19215$ &  $6.914$   \\   
			
			\multirow{3}{*}{(ImageNet)} & GSE
			& $98.4$ & $\boldsymbol{558}$  &   $1392$&        $\boldsymbol{5.137}$ &
			& $99.6$ & $2953$  &   $5830$ &        $\boldsymbol{5.933}$    \\

			& ATOS-GW
			& $\boldsymbol{100.0}$ & $666$  &  
			$\boldsymbol{998}$ &  $6.144$&
			& $\boldsymbol{100.0}$ & $\boldsymbol{1957}$  & 
			$\boldsymbol{2261}$ & $6.191$\\ 
			\cline{2-11}
			\multicolumn{11}{c}{ \footnotesize ATOS-GW($c: 0.02;  \lambda_s: 0.02,0.024$ )}\\
			\hline
			\multirow{3}{*}{ViT-s16}
			& FWnucl
			& $86.1$ & $16734$  &   $18956$ & $3.582$ &
			& $22.5$ & $16197$  &  $17586$ & $3.135$   \\   
			
			\multirow{3}{*}{(ImageNet)} & GSE
			& $\boldsymbol{100.0}$ & $\boldsymbol{374}$  &  $1107$ &        $\boldsymbol{1.424}$ &
			& $99.7$ & $\boldsymbol{2325}$  &$5113$ &        $\boldsymbol{2.601}$   \\   
			
			& ATOS-GW
			& $\boldsymbol{100.0}$ & $860$  &  
			$\boldsymbol{860}$ &   $3.264$&
			& $\boldsymbol{100.0}$ & $2644$  & 
			$\boldsymbol{2644}$ &$3.318$\\

			\cline{2-11}
			
			\multicolumn{11}{c}{ \footnotesize ATOS-GW($ c: 0.02;  \lambda_s: 0.02$ )}\\
			\hline
		\end{tabular}
	}
\end{table}
\begin{table}[htbp]
	\centering
	\caption{Comparison of the proposed imperceptible group-wise attack (ATOS-GW-$\ell_\infty$) to other methods. Common parameters are: $\mu: 0.2; s_\lambda: 1.5$.}
	\label{tab:group_imp}
	\resizebox{0.8\textwidth}{!}{
		\begin{tabular}{ccccccccccccc}
			\hline 
			Network & \multirow{2}{*}{Methods} & 
			\multicolumn{3}{c}{Untargeted Case}& \multicolumn{8}{c}{Worst Case} \\
			\cline{3-7} 
			\cline{9-13} 
			(Dataset)&  & ASR & NPP  & $\ell_\infty$  & $\ell_{2,0}$  & $T$  && ASR  & NPP  & $\ell_\infty$  &  $\ell_{2,0}$   & $T$  \\ \hline
			\multirow{3}{*}{CNN}
			
			& StrAttack
			& $\boldsymbol{100.0}$ & $170$  & $0.053$  &  $477$ &      $0.163$ &
			& $\boldsymbol{100.0}$ & $301$  & $0.096$  & $602$ &      $0.166$   \\   
			
			\multirow{3}{*}{(CIFAR-10)}& Homotopy
			& $63.0$ & $309$  & $\boldsymbol{0.029}$  &  $374$ & $198.906$ &
			& $51.0$ & $\boldsymbol{160}$  & $0.069$  & $483$ &      $581.288$   \\
			& ATOS-GW-$\ell_\infty$  
			& $\boldsymbol{100.0}$ & $\boldsymbol{101}$  & $0.033$  &  $\boldsymbol{329}$  &
			$\boldsymbol{0.131}$&
			& $\boldsymbol{100.0}$ & $172$  & $\boldsymbol{0.069}$   & $\boldsymbol{470}$ &  $\boldsymbol{0.132}$   \\ 
			\cline{2-13}
			
			\multicolumn{13}{c}{\footnotesize ATOS-GW-$\ell_\infty$($N_s: 5 ; c: 0.01; \lambda_s: 0.15,0.1;\lambda_\infty: 0.2; \sigma_m: 0.03,0.05$)}\\
			\hline
			\multirow{2}{*}{ResNet152}
			
			& StrAttack
			& $\boldsymbol{100.0}$ & $9324$  & $0.037$  &  $9324$ &  $11.582$ &
			& $97.9$ & $22894$  & $0.076$   & $27486$ &    $\boldsymbol{11.614}$   \\   
			\multirow{2}{*}{(ImageNet)} & ATOS-GW-$\ell_\infty$   
			& $\boldsymbol{100.0}$ & $\boldsymbol{1342}$  & $\boldsymbol{0.026}$  & 
			$\boldsymbol{1753}$ &  $\boldsymbol{6.450}$&
			& $\boldsymbol{100.0}$ & $\boldsymbol{3455}$  & $\boldsymbol{0.05}$   &
			$\boldsymbol{3455}$ &  $12.361$\\  
			\cline{2-13}
			
			\multicolumn{13}{c}{ \footnotesize ATOS-GW-$\ell_\infty$($N_s: 10 ; c: 0.015; \lambda_s: 0.15,0.1;\lambda_\infty: 0.2; \sigma_m: 0.02,0.04$)}\\
			\hline
			\multirow{2}{*}{ViT-s16}
			& StrAttack
			& $\boldsymbol{100.0}$ & $6111$  & $0.042$  &  $9581$ &        $\boldsymbol{4.919}$ &
			& $\boldsymbol{100.0}$ & $17672$  & $0.109$  & $23967$ &         $\boldsymbol{4.861}$   \\   			
			\multirow{2}{*}{(ImageNet)}& ATOS-GW-$\ell_\infty$   
			& $\boldsymbol{100.0}$ & $\boldsymbol{1249}$  & $\boldsymbol{0.029}$    & 
			$\boldsymbol{1715}$ & $5.740$&
			& $\boldsymbol{100.0}$ & $\boldsymbol{3829}$  & $\boldsymbol{0.066}$  &  
			$\boldsymbol{4835}$ &  $10.144$\\ 
			\cline{2-13}
			
			\multicolumn{13}{c}{ \footnotesize ATOS-GW-$\ell_\infty$($N_s: 10 ;c: 0.015,0.01; \lambda_s: 0.15,0.1;\lambda_\infty: 0.2,0.3;  \sigma_m: 0.02,0.04$)}\\
			\hline
		\end{tabular}
	}
\end{table}

\begin{figure}[htbp]
	\includegraphics[width=0.5\textwidth]{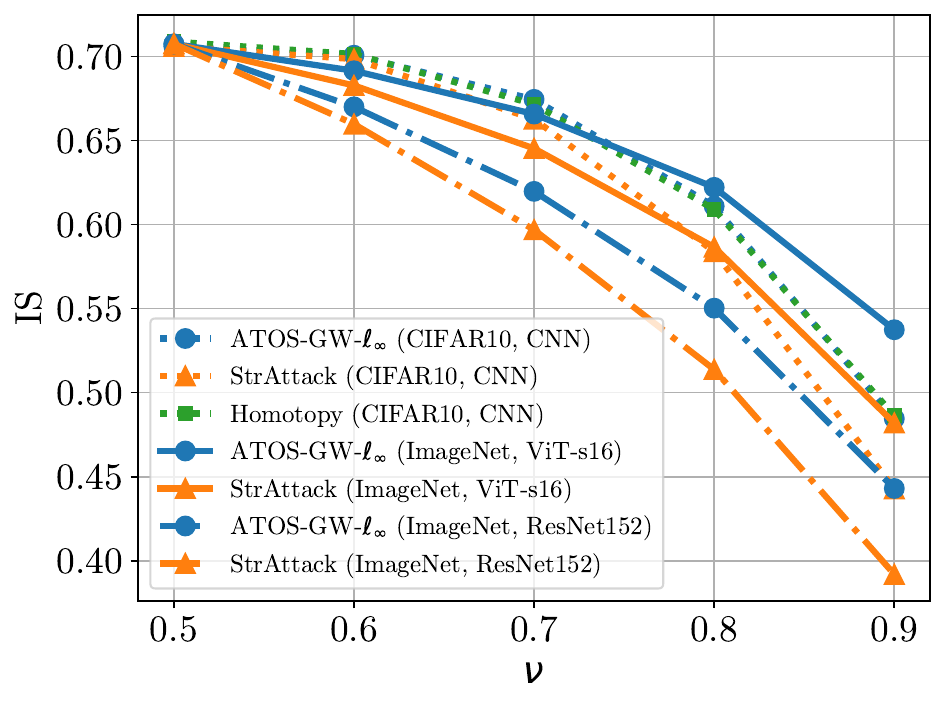}
	\centering
	\caption{IS comparison of imperceptible group-wise attacks for worst-case target versus threshold $v$ for binary ASM.}\label{fig:IS}
	\Description{A plot showing IS versus threshold for group-wise attacks in binary ASM.}
\end{figure}

\begin{figure*}[htbp]
	\includegraphics[width=0.9\textwidth]{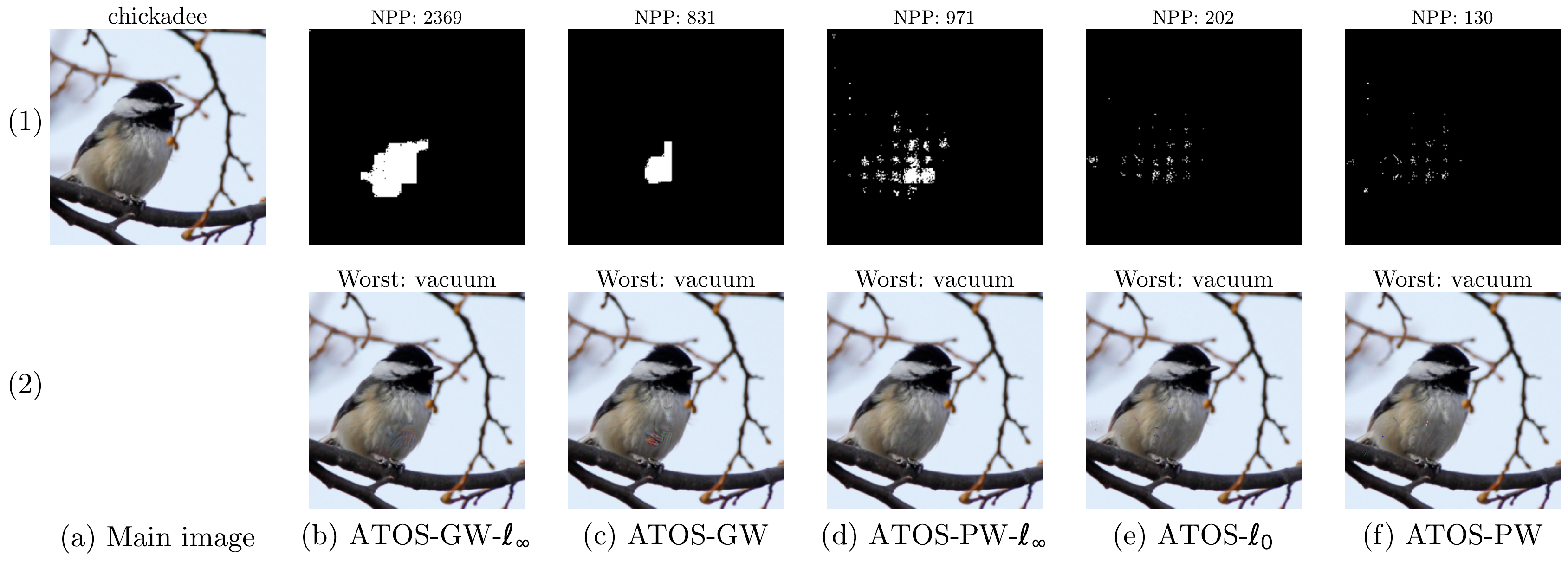}
	\centering
	\caption{Comparison of our different adversarial attacks on an ImageNet sample, targeting the ViT-S16 model. Column (a) shows an ImageNet sample labeled as "chickadee". Columns (b) to (f) present the results of our respective attacks. The binary mask in each column (row (1)) indicates the positions of perturbations with the number of perturbed pixels on its top, while the bottom part (row (2)) shows the adversarial example to misclassify the image as the worst-case label, "vacuum".}\label{fig:pos}
	\Description{A visual comparison of adversarial attacks applied to an ImageNet image of a chickadee. The figure shows the original image followed by several perturbed versions produced by different attack methods. Each column displays a binary mask of modified pixels and the resulting adversarial example, which misclassifies the image as a vacuum.}
	
\end{figure*}

\subsubsection{Group-Wise Attack}
Table {\ref{tab:group}} presents the results for the group-wise mode, which primarily aims to localize perturbations to specific regions of the image, thereby reducing their overall dispersion. The results indicate that only our ATOS attack achieves a 100\% ASR, except for the ViT-s16 model, where GSE achieves the same result. Consequently, our group-wise attack demonstrates superior performance in terms of NPP in most cases. Another important observation is that our group-wise attack results in the lowest number of perturbed blocks ($\ell_{2,0}$). The gap is considerable, particularly on the ImageNet dataset, indicating that our attack is less widespread and focuses on fewer blocks.

Table \ref{tab:group_imp} indicates the results for the group-wise attack with an invisibility constraint. Based on the mean of $\ell_\infty$ in the table for our attacks, which is below $0.03$ for untargeted cases and below $0.07$ for the worst cases, almost all perturbations from our attacks are imperceptible. Moreover, the table demonstrates that our attack achieves the best performance across all criteria, with an ASR of $100\%$.

Figure \ref{fig:IS} shows that among the imperceptible group-wise attacks, our ATOS attacks consistently exhibit better interpretability based on the IS metric across all datasets and models we examined. This indicates that ATOS-GW attacks target the image locations that are more critical for the model's decision-making for specific targets, particularly in comparison with the StrAttack. Although the IS of the Homotopy method is similar to ours, its low ASR and time-consuming nature made it impractical for us to experiment with this method on ImageNet (Table \ref{tab:group_imp}), rendering it inferior to our method. In terms of time complexity, all of our attacks demonstrate the best or near-best performance, consistently achieving an ASR of $100\%$, as shown in Tables~\ref{tab:sparse}, \ref{tab:invisible}, \ref{tab:group} and \ref{tab:group_imp}.

\begin{figure*}[htbp]
	\includegraphics[width=\textwidth]{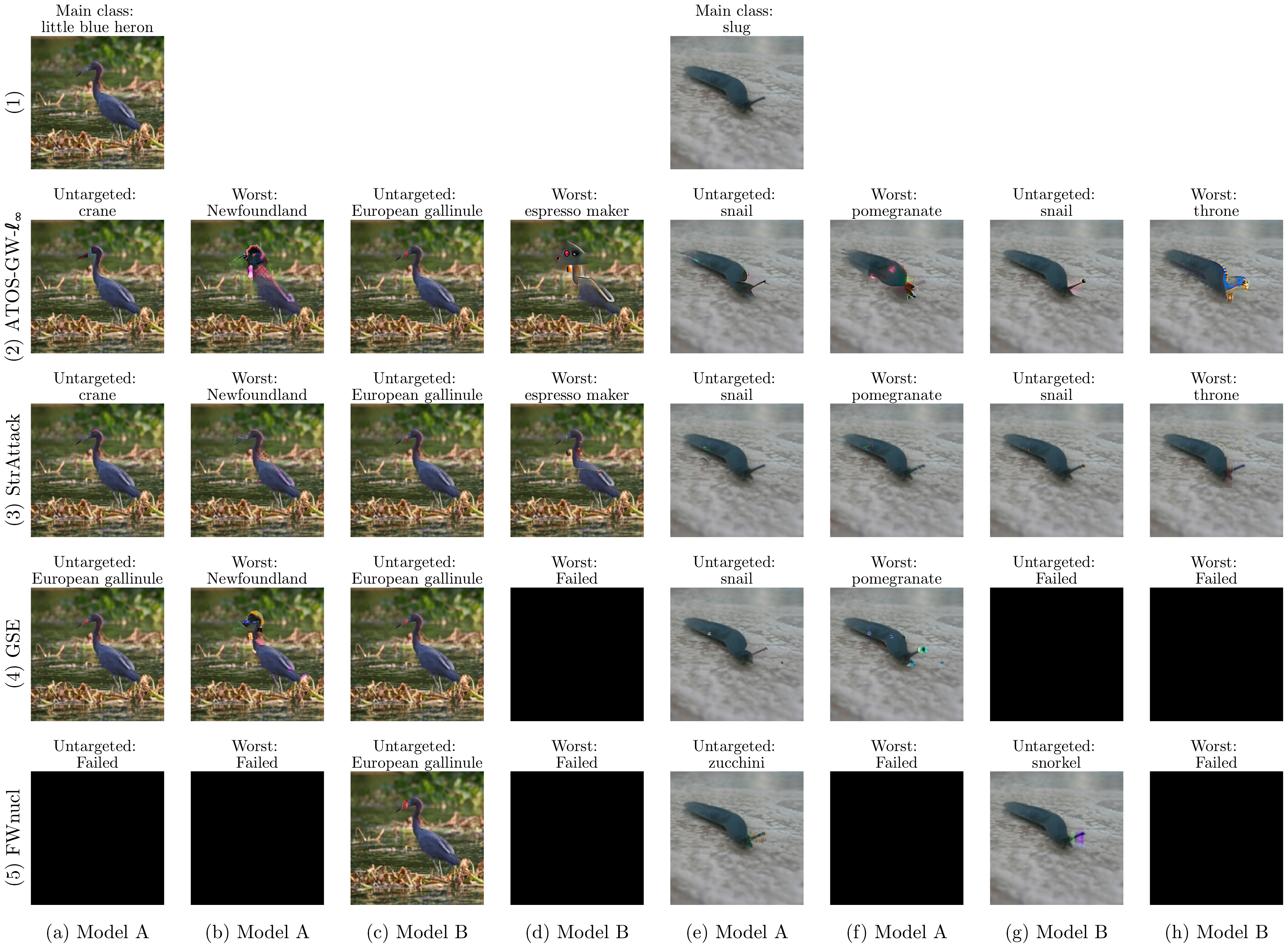}
	\centering
	\caption{Counterfactual explanations generated by ATOS-GW-$\ell_\infty$, StrAttack, GSE, and FWnucl (Rows (2) to (5), respectively) for two random samples from the ImageNet dataset in the robust networks of \cite{liu2024comprehensive} (Model A) and \cite{amini2024meansparse} (Model B). Row (1) contains the original images with their corresponding labels. Columns (a), (c), (e), and (g) correspond to counterfactual explanations for untargeted cases, while the remaining columns represent worst-case scenarios with their respective targeted labels.}
	\label{fig:pos2}
	\Description{A grid of counterfactual explanation images comparing multiple adversarial methods on two ImageNet samples across two robust models. The first row shows original labeled images, while subsequent rows show visual modifications produced by ATOS-GW-linf, StrAttack, GSE, and FWnucl. Alternating columns display untargeted and targeted counterfactual results for each method.}
\end{figure*}

Figure \ref{fig:pos} visualizes all versions of ATOS on the ViT-s16 model. It shows that even our non-group-wise perturbations follow the patch structure used by ViT models for learning. 
So far, since all the models we investigated lack robustness and primarily rely on non-robust features, the perturbations are partially located on the object, as shown in Figure \ref{fig:pos}. This indicates that the ATOS attacks are primarily concerned with how the model perceives the image, rather than being meaningful to the human eye. While this is advantageous for the purpose of attacks, it also demonstrates that these networks lack interpretability.

\subsubsection{Counterfactual Explainations}
Figure \ref{fig:pos2} demonstrates a case where group-wise attacks (except for Homotopy, which is not applicable) are generated on the robust networks of \cite{liu2024comprehensive} and \cite{amini2024meansparse}. 
According to the test results, none of the attacks successfully achieved invisibility, as expected. This is because we relaxed the constraints to maximize ASR. However, GSE and FWnucl do not always succeed in fooling the networks (as shown in Table \ref{tab:group}). In addition, the perturbations generated by FWnucl for these robust networks appear noise-like rather than semantically meaningful. 

StrAttack, like our ATOS attack, successfully misleads the models in all scenarios, but its perturbations are not always meaningful, as seen in the result for the image of a "slug". GSE generates more semantically relevant perturbations than StrAttack, but it sometimes fails to succeed. 

Ultimately, images produced by Algorithm \ref{alg:main} (ATOS-GW-$\ell_\infty$) consistently achieve a $100\%$ ASR while maintaining the semantics of the predicted class due to grouping. Regarding Figure \ref{fig:IS}, ATOS-GW-$\ell_\infty$ better aligns with the networks’ interpretations, generating more semantically meaningful perturbations than other methods (Figure \ref{fig:pos2}). Thus, we can interpret our attack on robust networks as a successful method for generating counterfactual explanations \cite{wachter2017counterfactual, goyal2019counterfactual, jeanneret2023adversarial}, unlike other attacks that do not always succeed. This approach can also be used to evaluate both model performance and interpretability.  

Based on the results of ATOS-GW-$\ell_\infty$ in Figure \ref{fig:pos2}, an interesting observation is that, in the untargeted case, the algorithm identifies the label that is semantically closest to the original image's label and generates a counterfactual explanation with that label.  
Furthermore, although the adversarial targets are the same in columns (e) and (g) ("snail"), the attack generates different results for each robust network. This occurs because the attack generates counterfactual explanations based on the perspective of each robust network.


\section{Conclusion} \label{c}
We have introduced ATOS algorithm for designing various modes of sparse structured adversarial attacks—namely element-wise, pixel-wise, and group-wise—based on the Overlapping Smoothed $\ell_0$ norm. Our algorithm guarantees convergence to a stationary point under mild conditions.
It also enables precise control over per-pixel perturbation, resulting in sparse and imperceptible adversarial modifications. The proposed methods achieve strong performance in reducing the number of perturbed pixels and groups and in increasing attack success rate, while maintaining efficient execution times. Furthermore, the group-wise mode serves as an effective tool for generating counterfactual explanations on robust networks.

\bibliographystyle{ACM-Reference-Format}
\bibliography{references}
\end{document}